\documentclass[runningheads]{llncs}
\usepackage{graphicx}
\usepackage{stmaryrd}
\usepackage{amsmath}
\usepackage{amssymb}
\usepackage{xcolor}
\usepackage{xspace}
\usepackage{listings}
\usepackage{lstautogobble}
\usepackage{todonotes}
\usepackage{mathtools}
\usepackage{mybnf}
\usepackage{tabularx}
\usepackage{mathpartir}
\usepackage{subcaption}
\usepackage{wrapfig}
\usepackage{booktabs}
\usepackage[inline]{enumitem}
\usepackage{hyperref}
\usepackage{cleveref}
\usepackage[firstpage]{draftwatermark}
\usepackage{orcidlink}
\usepackage{marvosym}
\usepackage{ifthen}

\newboolean{fullversion}
\setboolean{fullversion}{true}

\let\oldFootnote\footnote
\newcommand\nextToken\relax

\renewcommand\footnote[1]{%
    \oldFootnote{#1}\futurelet\nextToken\isFootnote}

\newcommand\isFootnote{%
    \ifx\footnote\nextToken\textsuperscript{,}\fi}

\SetWatermarkAngle{0}
\SetWatermarkText{\raisebox{13.3cm}{
\hspace{0.1cm}
\href{https://doi.org/10.5281/zenodo.15222546}{\includegraphics{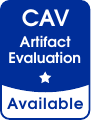}}
\hspace{8cm}
\includegraphics{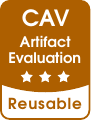}
}}

\usetikzlibrary{tikzmark,decorations,decorations.pathreplacing,arrows.meta,backgrounds,patterns}

\usetikzlibrary{calc}

\newcommand{\RR}{\mathbb{R}}
\newcommand{\TT}{\mathbb{T}}
\newcommand{\VV}{\mathbb{V}}

\newcommand\syn[1]{\texttt{\color{brown}{#1}}}
\newcommand{\srin}[0]{\mathit{SR}_\texttt{in}}
\newcommand{\srout}[0]{\mathit{SR}_\texttt{out}}
\newcommand{\sr}[0]{\srin \uplus \srout}
\newcommand{\world}[0]{\mathcal{W}}
\newcommand{\inputs}{\world_\mathit{in}}

\newcommand{\sliceFn}{\operatorname{slice}}
\newcommand{\nextDlFn}{\mathit{nextDl}}
\newcommand{\scheduleTy}{\mathit{Deadline}}
\newcommand{\memoryTy}{\mathit{Memory}}
\newcommand{\prefixTy}{\mathit{Prefix}}
\newcommand{\prefixMapTy}{\mathit{PrefixMap}}

\newcommand{\evalExp}[1][\p]{\Downarrow^{#1}_t}
\newcommand{\evalGuard}[1][\p]{\Downarrow^{#1}_{I,D,t}}
\newcommand{\evalStmt}[1][\p]{\evalStmtP{#1}{I,D,t}}
\newcommand{\evalStmtP}[2]{\Downarrow^{#1}_{#2}}

\newcommand{\translate}{\operatorname{translate}_\varphi}

\newcommand{\rtlola}{\mbox{RTLola}\xspace}
\newcommand{\p}{\mathit{inst}}
\newcommand{\instExp}{\mathit{inst}_{\mathit{exp}}}
\newcommand{\pre}{\mathit{pre}}

\newcommand{\streamir}{\mbox{StreamIR}\xspace}
\allowdisplaybreaks

\definecolor{bluekeywords}{rgb}{0.13, 0.13, 1}
\definecolor{greentypes}{rgb}{0, 0.5, 0}
\definecolor{orangecomments}{rgb}{1, 0.5, 0.1}
\definecolor{redstrings}{RGB}{171, 114, 2}
\definecolor{graynumbers}{rgb}{0.5, 0.5, 0.5}
\definecolor{goldcomments}{rgb}{0.6, 0.4, 0.08}

\lstdefinelanguage{Lola}{
  keywords=[0]{input, output, trigger, constant, import, spawn, eval, close, with, when, schedule},
  keywordstyle=[0]\bfseries\color{bluekeywords},
  keywords=[1]{if, then, else, aggregate, defaults, offset, by, or, to, sin, cos, abs, hold, over, using, min, max, in},
  keywords=[2]{Variable, String, Int, Int8, Int64, UInt, UInt8, UInt64, Bool, Float32, Float64, Float, @1Hz, @5Hz, @10Hz, @100mHz, @1kHz, @1min},
  keywordstyle=[2]\color{greentypes},
  	moredelim=[is][\textcolor{greentypes}]{|}{|},
    sensitive=false,
    comment=[l]{//},
    morecomment=[s]{/*}{*/},
    morestring=[b]',
    morestring=[b]"
}

\lstdefinelanguage{ir}{
  keywords={skip, shift, input, input?, spawn, eval, close, if, then, else, iterate, assign, schedule, dynamic, get, syn, const, self, window, global, local, get, ;, $\|$ |}
}

\lstdefinestyle{ir}{
  language=ir,
  basicstyle=\rmfamily\itshape,
  keywordstyle=\upshape\ttfamily\color{brown},
  alsoletter=?;,
  literate=
    {(}{{\normalfont(}}1
    {]}{{\normalfont)}}1
}

\begin{document}
\title{An Intermediate Program Representation for Optimizing Stream-Based Languages}
\titlerunning{Optimizing Stream-Based Languages}
\author{
  Jan Baumeister\textsuperscript{(\Letter)}\orcidlink{0000-0002-8891-7483}
  \and Arthur Correnson\orcidlink{0000-0003-2307-2296}
  \and \\ Bernd Finkbeiner\orcidlink{0000-0002-4280-8441}
  \and Frederik Scheerer\orcidlink{0009-0007-8115-0359}
}
\authorrunning{Baumeister et~al.}

\institute{CISPA Helmholtz Center for Information Security,\\Saarbrücken, Germany\newline
\email{\{jan.baumeister, arthur.correnson,\\finkbeiner, frederik.scheerer\}@cispa.de}}
\maketitle
\begin{abstract}
Stream-based runtime monitors are safety assurance tools that check at runtime whether the system’s behavior satisfies a formal specification.
Specifications consist of stream equations, which relate input streams, containing sensor readings and other incoming information, to output streams, representing filtered and aggregated data.
This paper presents a framework for the stream-based specification language \rtlola.
We introduce a new intermediate representation for stream-based languages, the \streamir, which, like the specification language, operates on streams of unbounded length; while the stream equations
are replaced by imperative programs.
We present a set of optimizations based on static analysis of the specification and have implemented an interpreter and a compiler for several target languages.
In our evaluation, we measure the performance of several real-world case studies.
The results show that the new \streamir framework reduces the runtime significantly
compared to the existing \rtlola interpreter.
We evaluate the effect of the optimizations and show that significant performance gains are possible beyond the optimizations of the target language’s compiler.
While our current implementation is limited to \rtlola, the \streamir is designed to accommodate other stream-based languages, enabling their interpretation and compilation into all available target languages.

\end{abstract}
\section{Introduction}
\label{sec:intro}

Frameworks for runtime monitoring must strike a careful balance between safety and performance. On the one hand, runtime monitors are responsible for raising alarms about potential risks and initiating mitigation procedures in real time; their safety and correctness are, therefore, critical. On the other hand, monitors are often run on platforms with very limited resources, such as the onboard computer of an unmanned aircraft, where efficiency is crucial.

To ensure the correctness of the monitor, a common approach is to follow the paradigm of \emph{specification-based monitoring}, where the monitor is not programmed in a programming language but rather based on a formal specification.
Many frameworks store the specification in memory and then \emph{interpret} the specification against the incoming data at runtime~\cite{faymonville2019streamlab}; this, however, causes substantially sub-optimal performance.
The alternative approach is to \emph{translate} the specification into a target programming language, such as Rust, and then rely on a standard compiler for the target language to produce executable code~\cite{finkbeiner2020verified,kallwies2022tessla,baumeister2019fpga}.
Even though translated monitors often perform significantly better than an interpreter, compilers for general-purpose programming languages do not optimize based on the specific properties of the specification language. This approach therefore still leaves potential for optimization on the table.

In this paper, we present an optimization framework for the class of stream-based specification languages~\cite{d2005lola,convent2018tessla,gorostiaga2018striver}.
Such specifications consist of stream equations, which relate input streams, containing sensor readings and other incoming information, to output streams, representing filtered and aggregated data. 
Stream-based monitoring has been applied to the safety assurance of cyber-physical systems, including monitoring of exhaust emissions in cars~\cite{DBLP:conf/tacas/BiewerFHKSS21} and the safety of unmanned aircraft~\cite{DBLP:conf/cav/BaumeisterFKLMST24,DBLP:conf/cav/BaumeisterFSST20}.
Stream-based languages have properties that are relevant for code optimization but are not found in general-purpose programming languages.
For example, the value of a particular expression in a given time step of a stream-based monitor is guaranteed to be the same, regardless of other stream evaluations in between.
As a result, if-statements can easily be combined or moved outside or inside other statements.

Our framework introduces a new intermediate representation, the \streamir, specifically designed to represent the relational stream-based specifications in an imperative form and enable optimizations before the code is translated to the target language.
We provide an interpreter and a compiler to Rust and Solidity.
The optimizations are implemented as a set of rewrite rules that exploit the specific properties of stream-based languages.
As the source language, our framework accepts \rtlola~\cite{faymonville2019streamlab,baumeister2024tutorial} specifications, which is a famous representative of the class of stream-based specification languages, but could easily be extended to support other stream-based languages like TeSSLa~\cite{convent2018tessla} or Striver~\cite{gorostiaga2018striver}. 

We evaluate our approach on a set of specifications from \rtlola case studies~\cite{DBLP:conf/cav/BaumeisterFKLMST24,baumeister2025fairness,baumeister2024tutorial,Scheerer/21}.
Interpreting the optimized \streamir code significantly outperforms the state-of-the-art \rtlola interpreter~\cite{faymonville2019streamlab}. This is perhaps not surprising, because the  \rtlola interpreter does not attempt to optimize the specification before the execution. However, our framework also outperforms the direct compilation of RTLola into Rust, demonstrating that the \streamir optimizations gain an additional speed-up on top of the optimizations already made by the Rust compiler.
The situation is similar in our experiments with monitoring smart contracts in Solidity.
Here, we assess our compiler by measuring the gas usage. The \streamir optimizations again achieve substantial performance gains beyond the optimizations of the Solidity compiler.

The rest of this paper is structured as follows:
After giving an overview on \rtlola in \Cref{sec:rtlola}, \Cref{sec:ir} introduces the intermediate representation, describes the translation from \rtlola specifiations to the StreamIR, and presents the StreamIR optimizations.
\Cref{sec:implementation} provides details about the implementation and the evaluation of our approach.

\subsubsection{Related Work.}
\label{sec:relatedWork}
Runtime monitoring is a dynamic verification approach that has been applied to a variety of domains~\cite{junges2021runtime,10.1007/978-3-031-37703-7_17}, and there are several monitoring tools and case studies~\cite{perez2020copilot,DBLP:conf/cav/BaumeisterFSST20,DBLP:conf/cav/BaumeisterFKLMST24}.
Monitoring with stream-based specifications was pioneered by Lola~\cite{d2005lola}, which was later extended to Lola2.0~\cite{faymonville2016stream} and \rtlola~\cite{faymonville2019streamlab,baumeister2024tutorial}.
Stream-based specification languages are related to synchronous programming languages such as Lustre~\cite{halbwachs1991synchronous}, Esterel~\cite{boussinot1991esterel} and Signal~\cite{gautier1987signal}.
In contrast to them, stream-based specifications are descriptive languages that use stream equations to describe temporal properties.
Other extensions of Lola are TeSSLa~\cite{convent2018tessla}, Striver~\cite{gorostiaga2018striver} or Copilot~\cite{DBLP:conf/fm/PerezGD24}.
Several compilers for stream-based languages exist, including compilers for Lola~\cite{finkbeiner2020verified} and TeSSLa~\cite{kallwies2022tessla} to Rust.
Other more specific compilers are a compilation~\cite{baumeister2019fpga} from \rtlola to the hardware description language VHDL~\cite{baumeister2019fpga} or a compilation from Copilot~\cite{pike2010copilot} to the Atom language, a domain-specific language for embedded hard real-time applications.
While some of these tools employ an intermediate representation (IR) to support different target languages, all are built for their specific specification language.

The concept of compiler construction through an IR is well-established.
The most well-known IR is LLVM, the basis for most modern compilers.
However, to the best of our knowledge, a general-purpose IR specifically designed for stream-based specifications does not yet exist.
Optimizations of stream-based specifications have been studied before.
In both RTLola~\cite{baumeister2020automatic} and TeSSLa~\cite{kallwies2022optimizing}, there exist approaches that describe optimizations specifically designed for their specification language.
This paper follows this idea, but, compared to the other approaches, the optimizations are defined on the more general StreamIR, which describes an imperative program, not a set of equations.
This approach is complementary to the optimizations based on the specification language.

\section{RTLola}
\label{sec:rtlola}

An \rtlola specification defines a set of streams, each representing an infinite sequence of values.
We differentate between \emph{input streams} which capture external data, and \emph{output streams} that perform computations based on current and past stream values.
Consider the specification in \Cref{fig:example:spec} which monitors a waypoint mission of an autonomous drone.
The drone is provided with new waypoints through an input stream and must reach each waypoint within 10~seconds.
We have exemplified an evaluation of the specifiation in \Cref{fig:example:spec_run}.
\begin{figure}[t]
\begin{lstlisting}
input pos : (UInt64, UInt64)
input wp : (UInt64, UInt64)
output moving
  eval |@pos| with pos != pos.offset(by: -1, or: pos)
output reached(current)
  spawn |@wp| with wp
  eval |@pos| with $\mathit{dist}$(pos, current) < $\varepsilon$
  close |@pos| when reached(current)
output missed_wp(current)
  spawn |@wp| with wp
  eval |@Local(10s)| with $\neg$reached(current).aggregate(over: 10s, using: $\exists$)
  close |@pos| when reached(current)
\end{lstlisting}
\caption{RTLola specification checking whether a drone reaches new waypoints within 10 seconds.}
\label{fig:example:spec}		
\end{figure}

The specification introduces two input streams: \lstinline!pos!, representing the drone's current position, and \lstinline!wp!, representing the position of new waypoints.
The output stream \lstinline!moving! determines whether the drone is currently in motion by comparing its current position to the previous one.
This is accomplished using an offset lookup -- a temporal operator in \rtlola that accesses past stream values.
The accesses are represented in \Cref{fig:example:spec_run} through a series of dashed arrows.
Since previous values don't exist initially, indicated by the question mark in the top left of the figure, the offset operator requires a default value.
The output stream \lstinline!reached! checks for each waypoint whether it was reached by the drone.
In contrast to the previous output stream, this stream is parameterized to describe a set of streams, called \emph{stream instances} -- each tracking a specific waypoint.

Instances are created using the \emph{spawn} clause, where the expression computes the \emph{parameter} of the instance.
Similarly, the \emph{close} clause defines when instances are removed.
In this example, a new instance is spawned for every new waypoint and closed once the corresponding waypoint is reached.
The \emph{evaluation} of parameterized streams is applied to every instance, in our example, if the distance between the current position and the waypoint, represented with the instance's parameter, is smaller than a threshold.
The figure depicts two instances of the \lstinline!reached! stream, each respectively monitoring whether waypoint $(5,3)$ and $(7,6)$ have been reached.
The output stream \lstinline!missed_wp! tracks whether a waypoint was missed -- that is, not reached within a time bound.
This stream is also parameterized over the waypoints, but in contrast to the previous streams, the evaluation is applied at a frequency.
More precicely, each stream instance is evaluated every 10~seconds and checks if the waypoint was not reached within the last 10~seconds.

\begin{figure}[t]
  \begin{center}
  \scalebox{1.0}{
  \begin{tikzpicture}[value/.style={draw,circle,fill=white,minimum width=5mm,inner sep=0.5mm},input/.style={value,rectangle,minimum height=4.5mm,minimum width=4.5mm}]
    \def\shift{-2mm}
    \node[anchor=east] (pos_label) {\lstinline!pos!};
    \node[below=8mm of pos_label.east,anchor=east] (moving_label) {\lstinline!moving!};
    \node[below=6mm of moving_label.east,anchor=east] (wp_label) {\lstinline!wp!};
    \node[below=6mm of wp_label.east,anchor=east] (reached1_label) {\lstinline!reached((5,3))!};
    \node[below=6mm of reached1_label.east,anchor=east] (missedwp1_label) {\lstinline!missed_wp((5,3))!};
    \node[below=6mm of missedwp1_label.east,anchor=east] (reached2_label) {\lstinline!reached((7,6))!};
    \node[below=6mm of reached2_label.east,anchor=east] (missedwp2_label) {\lstinline!missed_wp((7,6))!};

    \draw[{Bracket[width=3mm]}-] (pos_label.east) -- ++ (8.2,0);
    \draw[{Bracket[width=3mm]}-] (wp_label.east) -- ++ (8.2,0);
    \draw[{Bracket[width=3mm]}-] (moving_label.east) -- ++ (8.2,0);

    \foreach \i/\v [count=\xi] in {1/{(5,2)},4/{(5,2)},5/{(4,3)},5.8/{(4,3)},7.5/{(5,3)}} {
      \node[input,xshift=\shift] (pos\xi) at (\i, |- pos_label) {\scriptsize \v};
    }

    \node[value] (moving1) at (pos1 |- moving_label) {$\bot$};
    \node[value] (moving2) at (pos2 |- moving_label) {$\bot$};
    \node[value] (moving3) at (pos3 |- moving_label) {$\top$};
    \node[value] (moving4) at (pos4 |- moving_label) {$\bot$};
    \node[value] (moving5) at (pos5 |- moving_label) {$\top$};

    \draw[->,densely dashed,shorten >=0.3mm] (moving1) -- (pos1);
    \draw[->,densely dashed,shorten >=0.3mm] (moving2) -- (pos2);
    \draw[->,densely dashed,shorten >=0.3mm] (moving2) -- (pos1);
    \draw[->,densely dashed,shorten >=0.3mm] (moving3) -- (pos3);
    \draw[->,densely dashed,shorten >=0.3mm] (moving3) -- (pos2);
    \draw[->,densely dashed,shorten >=0.3mm] (moving4) -- (pos4);
    \draw[->,densely dashed,shorten >=0.3mm] (moving4) -- (pos3);
    \draw[->,densely dashed,shorten >=0.3mm] (moving5) -- (pos5);
    \draw[->,densely dashed,shorten >=0.3mm] (moving5) -- (pos4);

    \draw[->,densely dashed,shorten >=1.5mm] (moving1) -- ++(-5.5mm,4.5mm) node {?};
    
    \foreach \i/\v [count=\xi] in {2/{(5,3)},6.5/{(7,6)}} {
      \node[input,xshift=\shift] (wp\xi) at (\i, |- wp_label) {\scriptsize \v};
    }

    \draw[{Bracket[width=3mm]}-{Bracket}] (reached1_label -| wp1) -- (reached1_label -| pos5);
    \draw[{Bracket[width=3mm]}-] (reached2_label -| wp2) -- ($(reached2_label.east) + (8.2,0)$);

    \draw[{Bracket[width=3mm]}-{Bracket[width=3mm]}] (missedwp1_label -| wp1) -- (missedwp1_label -| pos5);
    \draw[{Bracket[width=3mm]}-] (missedwp2_label -| wp2) -- ($(missedwp2_label.east) + (8.2,0)$);

    \foreach \i in {pos_label,moving_label,wp_label,reached2_label,missedwp2_label} {
      \draw[dotted] (8.2, |- \i) -- ++(0.5,0);
    }

    \foreach \i in {1,2} {
      \node[value] (mweval\i) at ($(missedwp1_label -| wp1) + (\i*2.5, 0)$) {$\top$};
    }

    \draw[transform canvas={yshift=-2mm},decorate,decoration={brace,mirror,amplitude=2mm}] (missedwp1_label -| wp1) -- node[below=2mm] {\scriptsize 10 s} (mweval1);
    \draw[transform canvas={yshift=-2mm},decorate,decoration={brace,mirror,amplitude=2mm}] (mweval1) -- node[below=2mm] {\scriptsize 10 s} (mweval2);

    \node[value] (reached1_1) at (pos2 |- reached1_label) {$\bot$};
    \node[value] (reached1_2) at (pos3 |- reached1_label) {$\bot$};
    \node[value] (reached1_3) at (pos4 |- reached1_label) {$\bot$};
    \node[value] (reached1_4) at (pos5 |- reached1_label) {$\top$};
    \node[value] (reached2_1) at (pos5 |- reached2_label) {$\bot$};

    \coordinate (rwp) at ([xshift=2.5cm]missedwp2_label -| wp2);
    \begin{scope}[on background layer]
      \fill[fill=lightgray] ([xshift=-2.5cm]mweval1 |- reached1_label.center) -- (mweval1.center) -- (mweval1.center |- reached1_label.center);
      \fill[fill=lightgray] ([xshift=-2.5cm]mweval2 |- reached1_label.center) -- (mweval2.center) -- (mweval2.center |- reached1_label.center);

      \clip (0,0) rectangle (missedwp2_label -| 8.2,);
      \fill[fill=lightgray] ([xshift=-2.5cm]rwp |- reached2_label.center) -- (rwp.center) -- (rwp.center |- reached2_label.center);
    \end{scope}
    \begin{scope}
      \clip (missedwp2_label -| 8.2,) rectangle ++(0.45,2);
      \fill[pattern=north west lines, pattern color=lightgray] ([xshift=-2.5cm]rwp |- reached2_label.center) -- (rwp.center) -- (rwp.center |- reached2_label.center);
    \end{scope}
  \end{tikzpicture}}
  \end{center}

  \caption{An evaluation of the RTLola specification from \Cref{fig:example:spec}.}
  \label{fig:example:spec_run}
\end{figure}
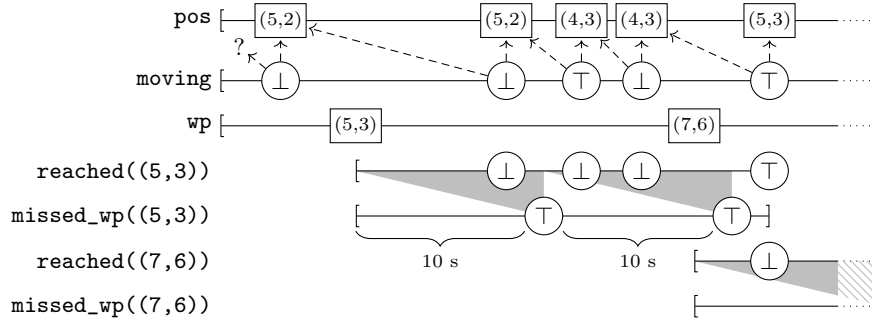

Note that for each clause -- spawn, eval, and close -- \rtlola has two types of filters: pacings and dynamic filter.
The \emph{pacing}, specified with an \lstinline!|@|!-symbol, describes whether the clause is event-based or periodic.
For \emph{event-based} clauses, the clause is evaluated when a new input arrives, such as the evaluation of \lstinline!moving!, which receives new values for every new value of \lstinline!pos!.
\emph{Periodic} clauses are evaluated with a fixed frequency, such as the evaluation of \lstinline!missed_wp! that is executed every 10 seconds.
\emph{Dynamic filters} are boolean stream expressions and follow after the \lstinline!when!-keyword.
They describe a condition based on stream values; here, we close a stream instance upon reaching the waypoint. 

From a specification, we can derive a \emph{Dependency Graph} describing the relation between streams.
Analysis on this graph returns the required memory and a partial order, indicating which evaluation steps -- called \emph{Tasks} -- have to be done in which order.
These tasks either spawn, shift, evaluate, or close stream instances and correspond to the clauses in the specification.
Here, \emph{shifting} an instance reserves space for a new value.
The order can be transformed into a list of \emph{layers}, where the layers are evaluated sequentially and all tasks of one layer in parallel.
For a more detailed explanation, we refer to the \rtlola tutorial~\cite{baumeister2024tutorial}.
\section{Stream Intermediate Representation (StreamIR)}
\label{sec:ir}

The \streamir, a new intermediate representation for stream-based specification languages, captures the imperative semantics of those languages and describes operations specific to real-time stream-based monitors, such as spawning, evaluating, and closing a stream instance.
Besides typical imperative statements, such as conditionals in the form of \syn{if}-\syn{then}-\syn{else}, it includes stream-based specific statements or expressions.
The \emph{guards} in conditions, for example, extend standard logical operators with constructs to check the presence of inputs to model event-based streams or constructs to reason about deadlines of periodic streams.
Using this representation, a StreamIR monitor $\syn{loop} \syn{\{} stmt \syn{\}}$ describes a statement that is executed for each input event.
\ifthenelse{\boolean{fullversion}}{
The complete syntax is defined in \Cref{app:syntax}.
}{
The complete syntax, semantics and optimizations are provided in the full version~\cite{full_version} of this paper.
}

\begin{example}
	The program in \Cref{fig:example:ir} describes a monitor for the specification from \Cref{fig:example:spec}.
	The program is split into three parts following the three pacings of the specification, \lstinline{|@pos|}, \lstinline{|@wp|}, and \lstinline{|@Local(10s)|}.
	First, the program handles the input \lstinline|wp|, and spawns the parameterized output streams.
	Then, it shifts and evaluates the \lstinline|pos| input stream and the non-parameterized \lstinline|moving| output stream before iterating over the instances of the \lstinline|reached| stream.
	Since the stream instances of \lstinline|missed_wp| are closed with the same condition as the \lstinline|reached| stream instances, the monitor handles the closing of these instances in parallel.
	The last guard iterates over the \lstinline|missed_wp| output stream to evaluate its instances.
  Given that the deadline is local, i.e. relative to the spawn of each instance, the guard checks the deadline of each instance individually.
\end{example}

\begin{figure}[t]
  \begin{lstlisting}[style=ir]
    if input? wp then
      shift wp ; input wp ; (spawn reached (syn wp] $\color{brown}\|$ spawn missed_wp (syn wp]] ;
    if input? pos then
      shift pos ; input pos ; shift moving ;
      eval moving (syn pos $\ne$ get pos -1 (syn pos]] ; 
      iterate reached
        shift reached ; eval reached dist(syn pos, self] < $\varepsilon$ ;
        if dynamic syn reached self then
          close reached $\color{brown}\|$ close missed_wp ;
    iterate missed_wp
      if local 10s missed_wp then
        shift missed_wp ; eval missed_wp $\neg$(window reached self 10s $\exists$]
    \end{lstlisting}
\caption{A StreamIR program of the example specification in \Cref{fig:example:spec}.}
\label{fig:example:ir}
\end{figure}

StreamIR monitors are defined over a sequence of \emph{Memories}.
Each memory contains the already computed values -- called \emph{prefix} -- of each stream instance and the next \emph{deadline} for periodic stream instances, the timestamp of the next stream evaluation according to the frequency.
We differentiate between global frequencies that start with the beginning of the monitor and local frequencies that start with the spawn of a stream instance.

Using the memory, we define inference rules to evaluate expressions and guards, i.e., $(M, \mathit{expr}) \evalExp \mathit{val}$, and statements, i.e., $(M, \mathit{stmt}) \evalStmt[inst] M'$.
A step of the monitor $((M,D),\mathit{stmt}) \Rightarrow_{I, t} (M', D')$ describes a sequence of statement executions, one execution for each passed deadline and one to process the input.
For a given input sequence, \emph{a monitor} then describes the sequence of memories according to the inference rules while starting with an initial memory.
\ifthenelse{\boolean{fullversion}}{
The list of all inference rules is given in \Cref{app:ir:sema}.
}{}

Some programs are ill-formed, i.e., the correctness of the program relies on constraints imposed on the input.
To avoid such programs, we only consider \emph{well-defined} programs, i.e., each infinite sequence of inputs has a unique output.
\subsection{Translation from \rtlola to StreamIR}
\label{sec:translation}
\newcommand{\graybox}{\tikz[baseline=1]\draw[gray,thick] (0,0) rectangle (0.2,0.2);}
\newcommand{\cyanbox}{\tikz[baseline=1]\draw[cyan,thick] (0,0) rectangle (0.2,0.2);}
\newcommand{\redbox}{\tikz[baseline=1]\draw[red,thick] (0,0) rectangle (0.2,0.2);}
\newcommand{\greenbox}{\tikz[baseline=1]\draw[green!50!black,thick] (0,0) rectangle (0.2,0.2);}
\newcommand{\bluebox}{\tikz[baseline=1]\draw[blue,thick] (0,0) rectangle (0.2,0.2);}

The translation from an \rtlola specification to the StreamIR follows the technique illustrated in \Cref{fig:translation:overview}.
Here, each stream in the specification is translated into a set of small StreamIR snippets.
Input streams are translated to a statement shifting the stream and writing the new value to memory.
Output streams are translated into four statements, one for each task -- spawn, shift, eval, and close.
The translated statements are then sorted according to their layers, where all statements in a layer are evaluated in parallel while the layers are concatenated sequentially.
The translation also generates an initial memory consisting of the prefixes and the deadlines.
The initial prefix of every input stream is assigned to an empty list whereas every output stream instance is marked as not-spawned.
The initial deadline assigns every global frequency to its first deadline.
\ifthenelse{\boolean{fullversion}}{%
The formal definition of the translation and the generation of the initial memory is defined in \Cref{app:translation}.
}{}

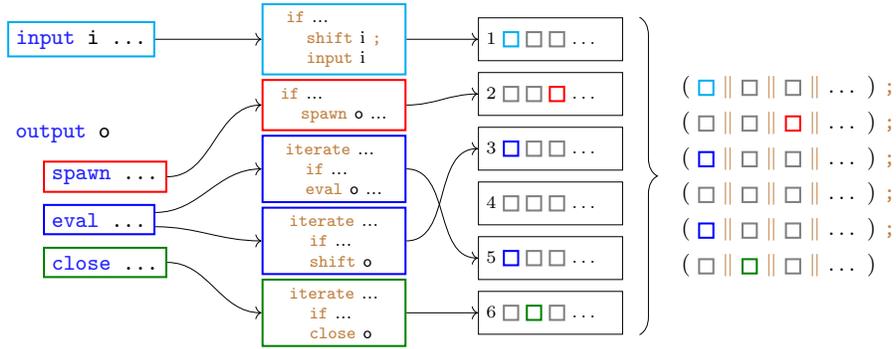
\begin{figure}[t]
	\centering
	\scalebox{0.95}{
	\begin{tikzpicture}
		\node[anchor=west] (o1) {\lstinline[basicstyle=\ttfamily]!output o!};
		\node[thick,draw=red,below=6mm of o1.west,anchor=west,xshift=5mm] (o2) {\lstinline[basicstyle=\ttfamily]!spawn ...!};
		\node[thick,draw=blue,below=6mm of o2.west,anchor=west] (o3) {\lstinline[basicstyle=\ttfamily]!eval ...!};
		\node[thick,draw=green!50!black,below=6mm of o3.west,anchor=west] (o4) {\lstinline[basicstyle=\ttfamily]!close ...!};
	
		\node[thick,draw=red,align=left,font=\scriptsize,above right=-0.2cm and 2cm of o1,minimum width=2cm] (ir1) {\syn{if}~...~\\\quad\syn{spawn}~\texttt{o}~...};
		\node[thick,draw=cyan,align=left,font=\scriptsize,above=0.5mm of ir1,minimum width=2cm] (ir0) {\syn{if}~...~\\\quad\syn{shift}~i~\syn{;}\\\quad\syn{input}~i};
		\node[thick,draw=blue,align=left,font=\scriptsize,below=0.5mm of ir1.south west,anchor=north west,minimum width=2cm] (ir2) {\syn{iterate}~...~\\\quad\syn{if}~...\\\quad\syn{eval}~\texttt{o}~...};
		\node[thick,draw=blue,align=left,font=\scriptsize,below=0.5mm of ir2.south west,anchor=north west,minimum width=2cm] (ir3) {\syn{iterate}~...~\\\quad\syn{if}~...\\\quad\syn{shift}~\texttt{o}};
		\node[thick,draw=green!50!black,align=left,font=\scriptsize,below=0.5mm of ir3.south west,anchor=north west,minimum width=2cm] (ir4) {\syn{iterate}~...~\\\quad\syn{if}~...\\\quad\syn{close}~\texttt{o}};
	
		\node[thick,draw=cyan,anchor=west] (i1) at (o1.west |- ir0) {\lstinline[basicstyle=\ttfamily]!input i ...!};
	
		\draw[->] (i1) to (ir0);
		\draw[->] (o2) to[out=0,in=180] (ir1);
		\draw[->] ([yshift=1mm]o3.east) to[out=0,in=180] (ir2);
		\draw[->] ([yshift=-1mm]o3.east) to[out=0,in=180] (ir3);
		\draw[->] (o4) to[out=0,in=180] (ir4);
	
		\node[font=\scriptsize,text width=1.8cm,align=left,minimum width=2cm,minimum height=0.6cm,draw,right=of ir0] (l1) {1~\cyanbox~\graybox~\graybox~\dots};
		\node[font=\scriptsize,text width=1.8cm,align=left,minimum width=2cm,minimum height=0.6cm,draw,right=of ir4] (l6) {6~\graybox~\greenbox~\graybox~\dots};
		\node[font=\scriptsize,text width=1.8cm,align=left,minimum width=2cm,minimum height=0.6cm,draw] (l2) at ($(l1)!{1/5}!(l6)$) {2~\graybox~\graybox~\redbox~\dots};
		\node[font=\scriptsize,text width=1.8cm,align=left,minimum width=2cm,minimum height=0.6cm,draw] at ($(l1)!{2/5}!(l6)$) (l3) {3~\bluebox~\graybox~\graybox~\dots};
		\node[font=\scriptsize,text width=1.8cm,align=left,minimum width=2cm,minimum height=0.6cm,draw] at ($(l1)!{3/5}!(l6)$) (l4) {4~\graybox~\graybox~\graybox~\dots};
		\node[font=\scriptsize,text width=1.8cm,align=left,minimum width=2cm,minimum height=0.6cm,draw] at ($(l1)!{4/5}!(l6)$) (l5) {5~\bluebox~\graybox~\graybox~\dots};
	
		\draw[->] (ir0) to[out=0,in=180] (l1);
		\draw[->] (ir1) to[out=0,in=180] (l2);
		\draw[->] (ir3) to[out=0,in=180] (l3);
		\draw[->] (ir2) to[out=0,in=180] (l5);
		\draw[->] (ir4) to[out=0,in=180] (l6);

		\draw[decorate,decoration={brace,amplitude=2.5mm,raise=2.3mm}] (l1.north east) to coordinate (c) (l6.south east);

		\node[right=7mm of c,anchor=west,yshift=1.25cm] (p1) {(~\cyanbox~\syn{$\|$}~\graybox~\syn{$\|$}~\graybox~\syn{$\|$}~\dots~)~\syn{;}};
		\node[below=5mm of p1.west,anchor=west] (p2) {(~\graybox~\syn{$\|$}~\graybox~\syn{$\|$}~\redbox~\syn{$\|$}~\dots~)~\syn{;}};
		\node[below=5mm of p2.west,anchor=west] (p3) {(~\bluebox~\syn{$\|$}~\graybox~\syn{$\|$}~\graybox~\syn{$\|$}~\dots~)~\syn{;}};
		\node[below=5mm of p3.west,anchor=west] (p4) {(~\graybox~\syn{$\|$}~\graybox~\syn{$\|$}~\graybox~\syn{$\|$}~\dots~)~\syn{;}};
		\node[below=5mm of p4.west,anchor=west] (p5) {(~\bluebox~\syn{$\|$}~\graybox~\syn{$\|$}~\graybox~\syn{$\|$}~\dots~)~\syn{;}};
		\node[below=5mm of p5.west,anchor=west] (p6) {(~\graybox~\syn{$\|$}~\greenbox~\syn{$\|$}~\graybox~\syn{$\|$}~\dots~)};
	\end{tikzpicture}}
	\caption{
		Overview of the translation from \rtlola specifications to StreamIR monitors.
	}
	\label{fig:translation:overview}
	\vspace{-0.5cm}
	\end{figure}

\subsection{Optimizations}
\label{sec:rewrite}
The translation of the monitors shown in \Cref{sec:translation} produces specific patterns, allowing for optimizations tailored to these patterns.
The \streamir representation of the monitor, properties of stream-based languages, and the analysis results enable the application of specialized rewriting rules.
These rewrite rules follow standard compiler techniques, such as combining if statements or loops.
However, since in a correctly constructed StreamIR program, every expression consistently evaluates to the same value regardless of its position, finding these patterns in the \streamir monitor is simpler than in a general-purpose language.
These compilers require complex program analysis which might miss the required assumptions to apply the optimizations.

\begin{example}
Consider the following snippets. The code on the left is produced by the translation from \rtlola and the code on the right is an optimized version:
\noindent
\begin{center}
\begin{minipage}{0.45\linewidth}
\begin{lstlisting}[style=ir, basicstyle=\rmfamily\itshape\scriptsize]
	iterate p	
		if schedule global 0.5
		  $\land$ dynamic self = syn a then
			eval p ...
	iterate q
		if schedule global 0.5
		  $\land$ dynamic self = syn a then
			eval q ...
\end{lstlisting}
\end{minipage}
\hspace{5mm}
\begin{minipage}{0.33\linewidth}
\begin{lstlisting}[style=ir, basicstyle=\rmfamily\itshape\scriptsize]
	if schedule global 0.5 then
		assign ( syn a]
			eval p ... ;
			eval q ...
\end{lstlisting}
\end{minipage}
\end{center}
First, we assume that the output streams $p$ and $q$ share the same spawn and close clauses, i.e., each clause has the same pacing, \lstinline|when|-condition, and \lstinline|with|-expression.
This property ensures that the iteration blocks iterate over the same parameters, so the first rewrite rule combines these iterations.
A second rule moves the first part of the guards outside the iteration since it is independent of the specific instance.
This optimization is possible because the global deadline is parameter-independent, whereas a local deadline would not permit such optimization.
Last, the remaining guard uniquely determines one instance, so the next rule replaces the iteration and the if statement with an assign statement directly calculating the parameter.
\ifthenelse{\boolean{fullversion}}{%
For the list of all rewriting rules, consider \Cref{app:rewriting}.
}{}
\end{example}

\paragraph{Memory.}
Stream-based specification languages operate over an infinite sequence of values.
However, an analysis can often determine static memory bounds per stream instance.
In \rtlola, all temporal accesses are bounded, allowing the dependency analysis to reveal a finite buffer length for all instances.
Based on a stream's structure, we can optimize the memory representation using a set of memory rewriting rules analogous to the rewriting rules from the previous paragraph.
They include optimizations for non-parameterized streams, instances with a buffer length of 1, and streams that are only accessed synchronously.

\paragraph{Symbolic Execution.}
\label{par:sym_execution}
The StreamIR enables a partial evaluation of guard conditions.
Code blocks surrounded by guards that are never satisfied can be removed, while guard checks can be eliminated if their conditions are never violated.
For instance, in \rtlola, the evaluation cycle is either completely event- or time-based.
By partially evaluating the \streamir, it is possible to create two versions: one containing event-based streams and one containing time-driven streams.
\section{Implementation \& Evaluation}

\subsubsection{Implementation.}
\label{sec:implementation}

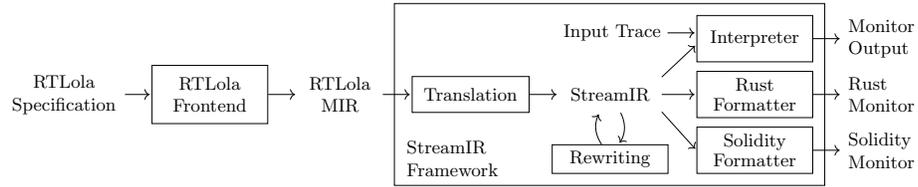
\begin{figure}[t]
	\centering
	\scalebox{0.77}{
\begin{tikzpicture}[formatter/.style={draw,align=center,anchor=west,minimum width=2cm},parser/.style={draw,minimum width=2cm},every path/.style={shorten >=1pt}]

	\node[parser,anchor=east,align=center,inner sep=2mm] (frontend) {RTLola\\Frontend};
	\node[align=center,inner sep=2mm,right=5mm of frontend] (MIR) {RTLola\\MIR};
	\node[parser,right=5mm of MIR,align=center,inner sep=2mm] (compiler) {Translation};
	\node[align=center,inner sep=2mm,right=5mm of compiler] (IR) {StreamIR};

	\node[formatter,minimum height=0.8cm,right=6mm of IR] (rust) {Rust\\[-0.5mm]Formatter};
	\node[formatter,above=1.5mm of rust,minimum height=0.8cm] (interpreter) {Interpreter};
	\node[formatter,below=1.5mm of rust,minimum height=0.8cm] (solidity) {Solidity\\[-0.5mm]Formatter};

	\node[left=5mm of interpreter,yshift=1mm] (trace) {Input Trace};
	\draw[->] (trace) -- (interpreter.west |- trace);

	\draw ([xshift=-3mm,yshift=2mm]compiler.west |- interpreter.north) rectangle ([xshift=2mm,yshift=-2mm]solidity.south east);
	\node[align=left,anchor=south west] at ([yshift=-1.5mm,xshift=-2mm]compiler.west |- solidity.south) {StreamIR\\Framework};

	\node[parser,anchor=south] (rewriting) at (IR |- solidity.south) {Rewriting};
	\draw[->] (IR) to[bend left] (rewriting);
	\draw[->] (rewriting) to[bend left] (IR);

	\draw[->] (IR.east) -- (rust.west);
	\draw[->] ([yshift=-3mm]IR.east) -- (solidity.west);
	\draw[->] ([yshift=3mm]IR.east) -- ([yshift=-1mm]interpreter.west);

	\draw[->] (frontend.east) --  (MIR);

	\draw[->] (rust) -- ++(1.5cm,0) node[anchor=west,align=left] {Rust\\Monitor};
	\draw[->] (solidity) -- ++(1.5cm,0) node[anchor=west,align=left] {Solidity\\Monitor};
	\draw[->] (interpreter) -- ++(1.5cm,0) node[anchor=west,align=left] {Monitor\\Output};
	\draw[->] (MIR) -- (compiler);
	\draw[->] (compiler) -- (IR);

	\draw[<-,shorten <=1pt] (frontend) -- ++(-1.5cm,0) node[anchor=east,align=center] {RTLola\\Specification};
\end{tikzpicture}}
\caption{Overview of the StreamIR framework}
\label{fig:compilerOverview}
\end{figure}
We extended the \rtlola framework~\cite{baumeister2024tutorial,faymonville2019streamlab} with a StreamIR-based interpreter and a compiler to Rust and Solidity.
We illustrate our setup in \Cref{fig:compilerOverview}.
Our approach uses the existing \rtlola frontend for parsing and analyzing specifications.
The RTLolaMIR includes all inferred information, such as pacing types or memory bounds.
We then apply the translation steps outlined in \Cref{sec:translation} as well as the rewriting rules introduced in \Cref{sec:rewrite}.
The library outputs the final StreamIR, which is then either interpreted or compiled.

The new \emph{interpretation} uses just-in-time (JIT) compilation for the complete \streamir while the existing \rtlola interpreter~\cite{baumeister2024tutorial} only uses JIT compilation for stream expressions.
This approach leverages the discussed optimizations, i.e., partial evaluation of the \streamir to decompose computation in time- and event-driven sections and rewrite rules to optimize control flow and memory.

For the \emph{compiler}, the StreamIR library provides a framework to translate the monitor into different target languages.
For each language, a formatter defines the translation of statements and expressions into the corresponding constructs of the target language.
Currently, we support the compilation to two programming languages used in different domains.
Our Rust formatter generates highly efficient Rust code that can run on microcontrollers.
Our second formatter compiles specifications to Solidity, a programming language for smart contracts.

Our artifact is available at Zenodo\footnote{\url{https://doi.org/10.5281/zenodo.15222546}} and our implementation is open-source and available on crates.io\footnote{\url{https://crates.io/crates/rtlola-streamir}}\footnote{\url{https://crates.io/crates/rtlola-streamir-interpreter}}\footnote{\url{https://crates.io/crates/rtlola2rust}}\footnote{\url{https://crates.io/crates/rtlola2solidity}}
and GitHub\footnote{\url{https://github.com/reactive-systems/rtlola-streamir}}.

\subsubsection{Evaluation.}
\label{sec:evaluation}
The evaluation was performed on a system with a 13th Gen Intel Core i7-1355U processor.
Rust code is compiled using rustc version 1.84.0 with the highest optimization level 3, and Solidity using solc version 0.8.24 optimized for 1000 runs.
We evaluate our implementation on three different benchmark sets and exclude the setup phases, such as specification analysis, for the runtime measurements.
For the optimizations, we apply the same set of rewriting rules for all specifiations.

\paragraph{Unmanned Aircraft.}
First, we evaluate our approach using two specifications from the field of monitoring unmanned aircraft.
The first specification is a geofencing application~\cite{DBLP:conf/cav/BaumeisterFKLMST24}, where a drone's position is continuously monitored to ensure it remains within a designated geofence.
This geofence is defined by a set of polygon lines in which the number of lines is proportional to the number of output streams.
The second specification is based on the \rtlola tutorial~\cite{baumeister2024tutorial}, observing the surrounding airspace of a drone to detect potential interference from other drones -- called intruders.
Unlike the first specification, which we evaluate across varying numbers of polygon lines, this benchmark varies the number of intruders while keeping the specification unchanged.
Here, the number of intruders corresponds to the number of stream instances.
For both specifications, we measure the execution time over 20 runs, each processing 10,000 input events.
We report the runtime of the old interpreter~\cite{faymonville2019streamlab}, the runtime of the new StreamIR interpreter with and without the optimizations, and the runtime of the optimized and unoptimized monitors compiled using the Rust formatter.

\begin{figure}[t]
	\begin{subfigure}[t]{0.485\linewidth}
		\includegraphics[width=\linewidth]{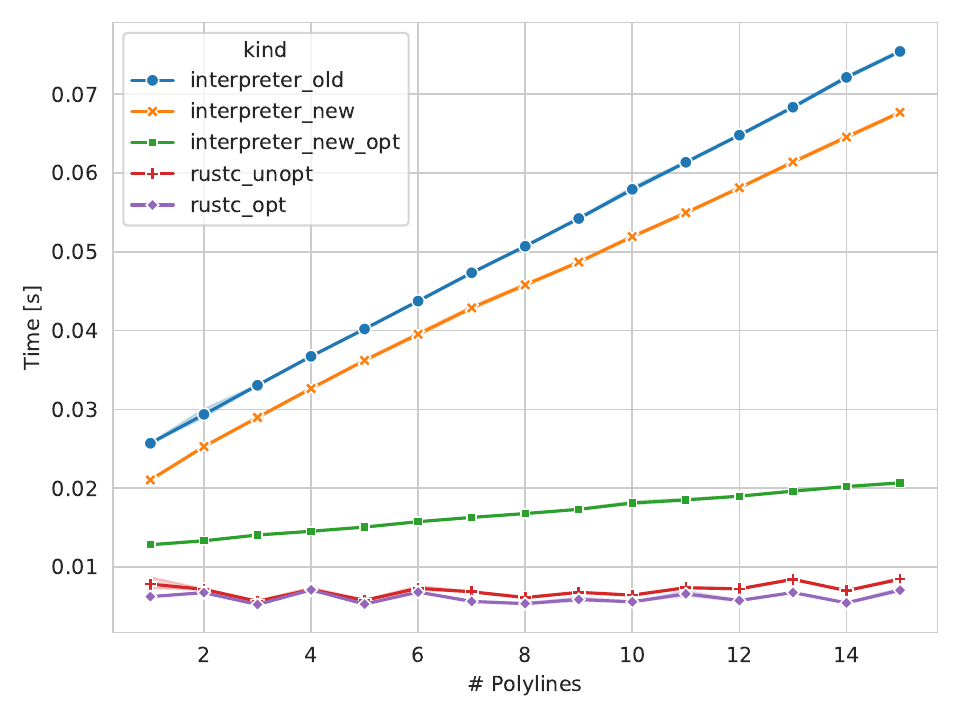}
		\caption{
			Geofence specification with increasing number of polygon lines.
		}
		\label{fig:eval:geofence}
	\end{subfigure}
	\hfill
	\begin{subfigure}[t]{0.485\linewidth}
		\includegraphics[width=\linewidth]{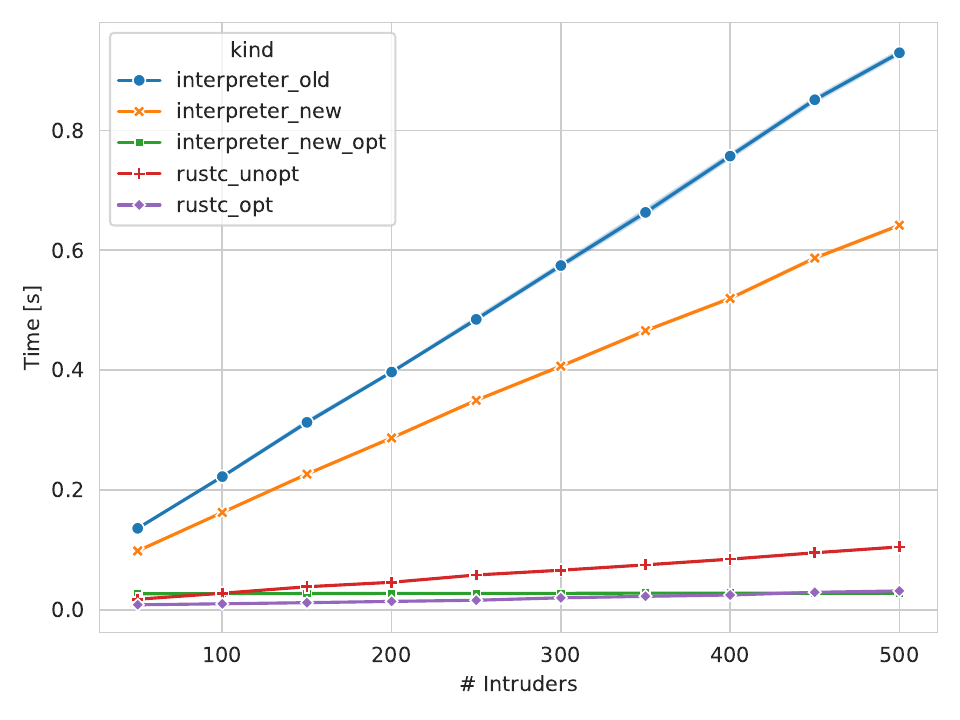}
		\caption{
			Intruder specification with increasing number of intruders.
		}
		\label{fig:eval:intruder}
	\end{subfigure}
	\caption{
		Runtime comparison between interpreters and compiled Rust monitors.
	}
	\label{fig:eval:interpreter}
\end{figure}

As seen in \Cref{fig:eval:interpreter}, the runtime of all monitors increases linearly because the number of stream evaluations increases either with the number of output streams or stream instances.
In general, a compilation is much faster than an interpretation.
When comparing the interpreter implementation, our new implementation outperforms the old one even without the optimizations by using JIT compilation for the complete StreamIR. 
After applying the StreamIR rewrite rules, the runtime is significantly reduced.
In contrast, the rewrite rules mostly do not affect the runtime for the geofence compilation to Rust, since the Rust compiler is capable of applying most optimizations on its own.
For the intruder specification however, StreamIR optimizations are able to significantly optimize the runtime of the compiled monitor.
The reason for this are optimizations of parameterized streams, which are easy in the StreamIR representation but hard for a Rust program in general.
\ifthenelse{\boolean{fullversion}}{%
The individual runtimes are presented in \Cref{app:runtimes} for the interested reader.
}{}

\setlength\intextsep{0pt}
\paragraph{Algorithmic Fairness.}
\begin{wraptable}[8]{r}{0pt}
	\setlength{\tabcolsep}{3pt}%
	\scalebox{0.85}{
	\begin{tabular}{lrr}
	\toprule
	& Unopt & Opt \\\midrule 
	Existing Interpreter & 17.258 & - \\
	StreamIR Interpreter & 12.889 & 7.943 \\
	Compilation & 3.157 & 1.443\\\bottomrule
	\end{tabular}
	}
	\caption{
		Runtime comparison of the fairness specification.
	}
	\label{fig:fair}
\end{wraptable}
The second benchmark utilizes the equalized-odds specification~\cite{baumeister2025fairness} for monitoring algorithmic fairness of the COMPAS tool.
COMPAS was developed by Northpointe~\cite{COMPAS} and predicts the recidivism risk of defendants.
A retrospective analysis~\cite{ProPublica16} showed the bias of the tool and our recent paper~\cite{baumeister2025fairness} demonstrates the use of stream-based monitoring to detect such bias.
\Cref{fig:fair} depicts the runtime on the COMPAS dataset, using the same evaluation setup described in the previous paragraph.
The table reports similar results to the previous benchmark.
Optimizations on the StreamIR reduce the runtime of the interpreter and compiler.

\paragraph{Smart Contracts.}
Efforts have explored monitoring smart contracts~\cite{azzopardi2018monitoring,abraham2019runtime} or expressing them in formal specification languages~\cite{finkbeiner2023reactive,mavridou2018designing,zupan2020secure}.
We extend this effort with the compilation to Solidity, so the monitor or smart contract itself can be expressed in \rtlola.
Since the interpretations can not be executed on the blockchain, our evaluation focuses on the performance benefits of the optimizations when compiling to Solidity.

We extended our framework with a function interface to compile functions describing the contract.
The functions set specific input stream values, call the monitor and return new values of output streams.
As not all input streams are set by a function, the compiler adapts the StreamIR using partial evaluation.
The evaluation uses \rtlola specifications that have previously been employed to monitor smart contracts~\cite{azzopardi2018monitoring,Scheerer/21,ERC20Reference} as well as existing Solidity contracts that we expressed in \rtlola~\cite{VotingReference,SimpleAuctionReference,CrowdFundingReference}.
\ifthenelse{\boolean{fullversion}}{
A description of the contracts and their functions can be found in \Cref{app:contracts}.
}{}
For the comparison, we assess the required gas, the transaction fee for calling functions on the blockchain, for contract deployment and the average gas consumption per function call over 1000 calls.

\begin{figure}[t]
	\begin{subfigure}[t]{0.49\linewidth}
	\includegraphics[width=\linewidth]{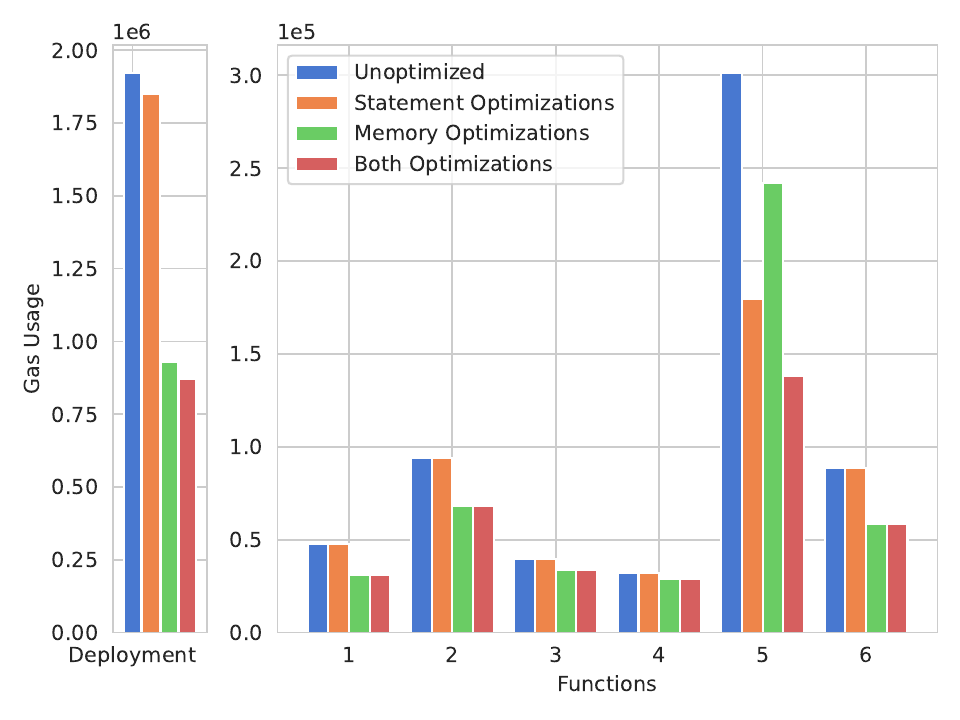}
	\caption{
		Order contract
	}
	\label{fig:eval:orders}
	\end{subfigure}\hfill
	\begin{subfigure}[t]{0.49\linewidth}
	\includegraphics[width=\linewidth]{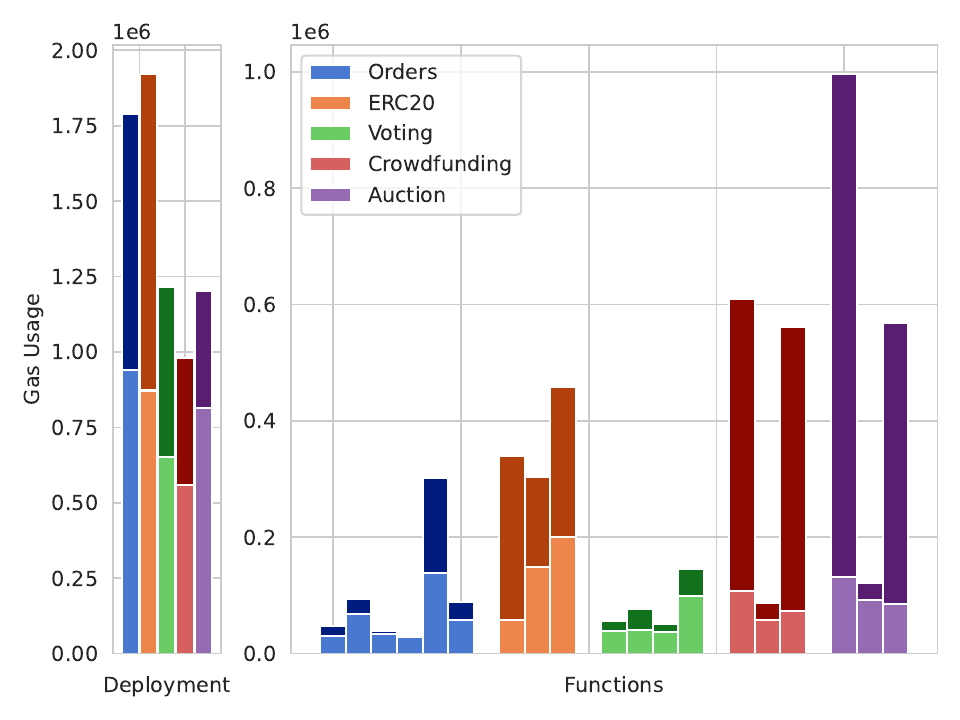}
	\caption{
		All contracts
	}
	\label{fig:eval:all}
	\end{subfigure}
\caption{
	Gas usage comparison of the Solidity monitors.
}
\label{fig:eval:solidity}
\end{figure}

\Cref{fig:eval:orders} displays the average gas consumption per function call when monitoring an ordering system, along with the gas costs for the contracts deployment.
The functions allow the interaction with the contract, for example to place a new order, mark an order as delivered or request termination of the contract.
The cost of the six functions and the deployment are each represented with four bars, where each bar represents the contract \begin{enumerate*}
	\item without any optimizations
	\item with statement optimizations
	\item with memory optimizations, and
	\item with both optimizations.
\end{enumerate*}
The results demonstrate that the deployment costs mostly are reduced by memory optimizations minimizing the use of global variables.
Furthermore, the gas reduction of some function calls relies on memory optimizations.
By reducing reliance on global storage, these optimizations yield substantial savings in deployment and runtime.
Meanwhile, StreamIR optimizations mainly improve function execution costs, especially for functions that update parameterized streams, such as the fifth function for placing a new order.
These functions benefit significantly from the iterate-assign optimization, which eliminates unnecessary iterations and removes the need to store activate parameters.

\Cref{fig:eval:all} gives an overview of the gas usage across a number of different contracts.
In this figure, each bar corresponds to a function in a contract.
The top half of the bar illustrates the gas costs for the StreamIR-unoptimized contract, while the lower half shows the reduction obtained by enabling all StreamIR optimizations.
For example, the dark blue bars in \Cref{fig:eval:all} correspond with the blue bars in \Cref{fig:eval:orders} and the light blue corresponds with the red bars.
\Cref{fig:eval:all} demonstrate that our optimizations reduce gas usage across all specifications.

Beyond gas usage reduction, the StreamIR optimizations address a fundamental issue of contracts with parameterized output streams.
Iteration introduces unbounded gas usage -- a highly undesirable property for smart contracts.
By applying the optimization that replaces iterations with assignments -- feasible for all evaluated specifications -- gas usage becomes bounded again.
\section{Conclusion}
\label{sec:conclusion}
We have presented a new framework for the stream-based specification language \rtlola, with an interpretation and a compilation to Rust and Solidity.
Our framework uses a new intermediate representation designed for stream-based specifications.
Further, we introduced rewrite rules for this representation to optimize the resulting monitors.
Even without optimizations, our experiments show that the new interpreter and compiler outperform the existing \rtlola interpreter.
Additionally, we show the impact of the rewriting rules optimizing the runtime of the monitor.
Especially in specifications with parameterized streams, our optimizations speed up the compiled monitor significantly.
For the interpretation, our optimizations also have a huge impact for specifications without parameterized streams.

\subsubsection{Acknowledgments.} This work was partially supported by the Aviation Research Program LuFo of the German Federal Ministry for Economic Affairs and Energy as part of "Volocopter Sicherheitstechnologie zur robusten eVTOL Flugzustandsabsicherung durch formales Monitoring" (No.~20Q1963C), by the German Research Foundation (DFG) as part of TRR 248 (No.~389792660), and by the European Research Council (ERC) Grant HYPER (No.~101055412).

\subsubsection{Disclosure of Interests.}
The authors have no competing interests to declare that are relevant to the content of this article.

\bibliographystyle{splncs04}
\bibliography{references.bib}

\ifthenelse{\boolean{fullversion}}{
\clearpage
\appendix

\section{Syntax}
\label{app:syntax}
\begin{figure}[h]
	\scriptsize
    \begin{bnfgrammar}
      stmt : \text{Statements \qquad} ::= \syn{skip} \textbar~\syn{shift} sr \textbar~\syn{input} sr 
	  | \syn{spawn} sr expr \textbar~\syn{eval} sr expr \textbar~\syn{close} sr
      | stmt \syn{;} stmt \textbar~stmt \syn{$\|$} stmt
      | \syn{if} guard \syn{then} stmt \syn{else} stmt
	  | \syn{iterate} sr stmt \textbar~\syn{assign} expr stmt
      ;;
      guard : \text{Conditions} ::= \syn{input?} sr \textbar~\syn{schedule} freq \textbar~\syn{dynamic} expr
	  | guard ( \syn{$\lor$} \(\vert\) \syn{$\land$} ) guard
      ;;
      expr : \text{Expressions \qquad} ::= \syn{const} c \textbar~\syn{get} sr expr off expr \textbar~\syn{syn} sr expr
      | expr ( \syn{+} \(\vert\) \syn{-} \(\vert\) \syn{*} \(\vert\) \dots ) expr \textbar~\syn{self}
	  | \syn{window} sr expr dur f \textbar~...
      ;;
	freq: \text{Frequencies} ::= \syn{global} dur \textbar~\syn{local} dur sr;;
	dur : \text{Durations} ::= $\RR^{+}$;;
	f : \text{Aggregations} ::= $\VV^* \rightarrow \VV$;;
	c : \text{Constants} ::= $\VV$
    \end{bnfgrammar}
    \caption{Syntax of StreamIR}%
    \label{fig:imp_syntax}
\end{figure}
\begin{definition}[StreamIR Monitors]
	A monitor $\phi = \syn{loop} \syn{\{} stmt \syn{\}}$ is an infinite loop over the statement $stmt$.
\end{definition}
\begin{remark}
	For readability, we use $\syn{if}~c~\syn{then}~stmt$ if the alternative is followed by the $\syn{skip}$ statement and removed the $\top$ placeholder for non-parameterized stream accesses in the paper.
\end{remark}
\section{Semantics}
\label{app:ir:sema}
\begin{definition}[Memory]
	A memory $M$ is defined as a pair of a \emph{$\prefixMapTy$} and a \emph{$\scheduleTy$}.
	\begin{align*}
		\prefixTy &:= (\TT, \VV)^*\\
		\prefixMapTy &:= \sr \rightarrow \VV \rightarrow (\prefixTy \cup \bot)\\
		\mathit{GlobalDeadline} &:= \mathbb{F} \rightarrow \RR_\bot\\
		\mathit{LocalDeadline} &:= (\mathbb{F} \times \sr \times \VV) \rightarrow \RR_\bot\\
		\scheduleTy &:=  \mathit{GlobalDeadline} \cup \mathit{LocalDeadline}\\
		\memoryTy &:= \prefixMapTy \times \scheduleTy
	\end{align*}
\end{definition}
For readability, $M[\mathit{sr}][\mathit{inst}]$ defines the lookup to the prefix of a stream instance in the prefixmap $M$.
The list notation $\mathit{pre} = \mathit{pre'} \cdot v$ splits the non-empty prefix $\mathit{pre}$ into the last element $v$ and the rest of the prefix $\mathit{pre'}$, $|\mathit{pre}|$ returns the length of the prefix, and $\mathit{pre}[n]$ returns the n-th element in the prefix.
For the aggregation functions, we use $\mathit{slice}$ to get the slice of a prefix:
\begin{definition}[Slice]
	Given a prefix $P$ and a timestamt $t$, the \emph{slice} of $P$ returns all values computed later than $t$:
	\begin{align*}
		\sliceFn&: \prefixTy \cup \bot \times \RR \rightarrow \VV^*\\
		\sliceFn(v, t) &= v \qquad\text{if } v \in \{ \bot, \epsilon\}\\
		\sliceFn(xs \cdot (v, t'), t) &= \begin{cases}
			\sliceFn(xs, t) \cdot v & \text{if } t \le t'\\
			\epsilon & \text{otherwise}
		\end{cases}
	\end{align*}
\end{definition}
\begin{figure}[h]
	\scalebox{0.8}{
	\begin{mathpar}
		\inferrule[eval-self]{ \p \ne \bot }{(M, \syn{self}) \evalExp \p}\qquad
		\inferrule[eval-syn]{ (M, \instExp) \evalExp \p' \\M[sr][\p'] = \pre \cdot (t, val) }{(M, \syn{syn}~\mathit{sr}~\instExp) \evalExp val}\\
		\inferrule[eval-get]{ (M, \instExp) \evalExp \p' \\ M[sr][\p'] = \pre\\\\|\pre| > \mathit{off} \\\pre[|\pre| - \mathit{off}] = (t',val) }{ (M, \syn{get}~\mathit{sr}~\instExp~\mathit{off}~\mathit{dft}) \evalExp val }\quad
		\inferrule[eval-get-dft]{ (M, \instExp) \evalExp \p' \\ M[sr][\p'] = \pre\\\\|\pre| \le \mathit{off} \\(M, dft) \evalExp val }{ (M, \syn{get}~\mathit{sr}~\instExp~\mathit{off}~\mathit{dft}) \evalExp val }\\
		\inferrule[eval-const]{c \in \mathbb{V} }{(M, \syn{const}~c) \evalExp c}\qquad
		\inferrule[eval-bin-op]{ (M, exp_1) \evalExp val_1\\(M, exp_2) \evalExp val_2 \\ \circ \in \{+, -, *, \dots\}}{ (M, exp_1 \ \syn{$\circ$} \ exp_2) \evalExp val_1 \circ val_2 }\\
		\inferrule[eval-window]{ (M, \instExp) \evalExp \p' \\ M[sr][\p'] = \pre \\ f(\sliceFn(\pre, t - dur)) = val}{ (M, \syn{window}~\mathit{sr}~\instExp~\mathit{dur}~\mathit{f}) \evalExp val}
	  \end{mathpar}}
	\caption{Operational semantics of stream expressions}
	\label{fig:inf-rules:expr}
\end{figure}
\begin{definition}[Evaluation of expressions]
	Given a prefixmap $M$, a stream expression $\mathit{expr}$, a current instance $\mathit{inst}$ and a timepoint $t$, the expression evaluation $(M, \mathit{expr}) \Downarrow^{\mathit{inst}}_{t} \mathit{val}$ is described by the inference rules in \Cref{fig:inf-rules:expr}.
\end{definition}
\begin{figure}[h]
	\scalebox{0.8}{
	\begin{mathpar}
		\inferrule[eval-input-guard]{I(sr) = v}{(M,\syn{input?}~sr) \evalGuard v \ne \bot}\qquad
		\inferrule[eval-bin-guard]{(M,exp_1) \evalGuard val_1\\(M, exp_2) \evalGuard val_2\\ \circ \in \{\land, \lor\}}{(M, exp_1 \ \syn{$\circ$} \ exp_2) \evalGuard val_1 \circ val_2}\\
		\inferrule[eval-schedule-global]{ D(freq) = t'}{(M, \syn{schedule}~\syn{global}~freq) \evalGuard t'= t}\qquad
		\inferrule[eval-schedule-local]{ D(freq, sr, inst) = t'}{(M, \syn{schedule}~\syn{local}~freq~sr) \evalGuard t' = t}\\
		\inferrule[eval-dynamic]{ (M,exp) \evalExp val\\ val \in \mathbb{B}}{(M,\syn{dynamic}~exp) \evalGuard val}\qquad
	\end{mathpar}}
	\caption{Operational semantics of guards}
	\label{fig:inf-rules:guards}
\end{figure}
\begin{definition}[Evaluation of guards]
	Given a prefixmap $M$, a guard $g$, a current instance $\mathit{inst}$, an input $I$, a deadline $D$ and a timepoint $t$, the evaluation of a guard $(M, \mathit{g}) \Downarrow^{\mathit{inst}}_{I, D, t} \mathit{val}$ is described by the inference rules in \Cref{fig:inf-rules:guards}.
\end{definition}

  \begin{definition}[Update]
	Given a prefixmap $M$, a stream reference $\mathit{sr}$, and an instance $\mathit{inst}$, an \emph{update} $M[\mathit{sr}][\mathit{inst} \leftarrow v]$ is defined as
	\begin{align*}
		&\cdot[\cdot][\cdot \leftarrow \cdot] : (\prefixMapTy \times \sr \times \VV \times \VV) \rightarrow \prefixMapTy\\
		&M[\mathit{sr}][\mathit{inst} \leftarrow v] = \lambda \mathit{sr'}, \mathit{inst'} \begin{cases}
			v & \text{if } \mathit{sr} = \mathit{sr}' \land \mathit{inst} = \mathit{inst'}\\
			M[\mathit{sr'}][\mathit{inst'}] & \text{otherwise}
		\end{cases}
	\end{align*}
\end{definition}

\begin{figure}[h]
	\scalebox{0.8}{
	\begin{mathpar}
		\inferrule[exec-skip]{ }{(M,\syn{skip}) \evalStmt M}\qquad
		\inferrule[exec-shift]{ }{(M,\syn{shift}~\mathit{sr}) \evalStmt M[sr][inst \leftarrow M[sr][inst]\cdot \bot]}\\
		\inferrule[exec-seq]{ (M_1, p_1) \evalStmt M_2\\(M_2, p_2) \evalStmt M_3}{(M_1, p_1 \syn{;} \ p_2) \evalStmt M_3}\qquad
		\inferrule[exec-parallel]{ (M, p_1 \syn{;} \ p2) \evalStmt M'\\(M, p_2 \syn{;} \ p_1) \evalStmt M'}{(M, p_1 \mathbin{\syn{$\|$}} p_2) \evalStmt M'}\\
		\inferrule[exec-input]{I[sr] = val\\val \ne \bot\\M[sr][\top] = str \cdot \bot }{ (M, \syn{input} \ sr) \evalStmt M[sr][\top \gets str \cdot (t, val) ]}\\
		\inferrule[exec-eval]{(M, exp) \evalExp val \\M[sr][inst] = str \cdot \bot}{ (M, \syn{eval}~sr~exp) \evalStmt M[sr][inst \gets str \cdot (t,val) ]}\qquad
		\inferrule[exec-close]{ }{ (M,\syn{close}~sr) \evalStmt M[sr][inst \gets \bot ]}\\
		\inferrule[exec-spawn-new]{(M, \mathit{inst}_\mathit{exp}) \evalExp \mathit{inst'}\\ M[sr][\mathit{inst'}] = \bot }{ (M, \syn{spawn}~sr~\mathit{inst}_\mathit{exp}) \evalStmt M[sr][inst' \gets \epsilon ]}\quad
		\inferrule[exec-spawn-exists]{(M, \mathit{inst}_\mathit{exp}) \evalExp \mathit{inst'}\\M[sr][inst'] \ne \bot }{ (M, \syn{spawn}~sr~\mathit{inst}_\mathit{exp}) \evalStmt M}\\
		\inferrule[exec-if-true]{(M,stmt) \evalStmt M'\\(M, g) \evalExp true}{(M, \syn{if} \ g \ \syn{then} \ stmt \ \syn{else} \ stmt') \evalStmt M'}\qquad
    	\inferrule[exec-if-false]{(M,stmt') \evalStmt M'\\ (M, g) \evalExp false}{(M, \syn{if} \ g \ \syn{then} \ stmt \ \syn{else} \ stmt') \evalStmt M'}\\
		\inferrule[exec-iterate]{ \{p_1, \dots, p_n\} = \{p \mathbin{|} M_0[sr][p] \neq \bot \}\\\\ (M_0, stmt) \evalStmt[p_1] M_1\quad \dots \quad(M_{n-1}, stmt) \evalStmt[p_n] M_n}{(M_0, \syn{iterate}~sr~stmt) \evalStmt[\top] M_n }\quad
		\inferrule[exec-assign]{(M, \instExp) \evalExp[\top] inst\\(M, stmt) \evalStmt[inst] M'}{(M, \syn{assign}~\instExp~stmt) \evalStmt[\top] M'}
	\end{mathpar}}
	\caption{Operational semantics of statements}
	\label{fig:inf-rules:stmts}	
\end{figure}

\begin{definition}[Evaluation of statements]
	Given a prefixmap $M$, a statement $\mathit{stmt}$, a current instance $\mathit{inst}$, an input $I$, a deadline $D$ and a timepoint $t$, the evaluation of statement $(M, \mathit{stmt}) \Downarrow^{\mathit{inst}}_{I, D, t} M'$ returns a new prefixmap $M'$ following the inference rules in \Cref{fig:inf-rules:stmts}.
\end{definition}

\begin{definition}[Next Deadline]
	Given a deadline $D$, a timepoint $t$, and the prefixmaps $M$ and $M'$, the \emph{deadline update} is defined as
	\begin{align*}
		\nextDlFn&: \scheduleTy \times \RR \times \prefixMapTy \times \prefixMapTy  \rightarrow \scheduleTy\\
		\nextDlFn(D, t, M, M') &= \mathit{nextGlobalDl}(D, t) \cup \mathit{nextLocalDl}(D, t, M, M')\\
		\mathit{nextGlobalDl}(D,t) &= \lambda f. \begin{cases} 
			t + f &\text{if } D(f) = t\\
			D(f) &\text{otherwise}
		\end{cases} \qquad
		\mathit{nextLocalDl}(D,t,M,M') =\\ &\hspace{-1.6cm}\lambda f, sr, inst. \begin{cases}
		t + f &\text{if } D(f, sr, inst) = t\\
		t + f &\text{if } M[sr][inst] = \bot \land M'[sr][inst] \ne \bot\\
		\bot & \text{if } M[sr][inst] \ne \bot \land M'[sr][inst] = \bot\\
		D(f, sr, inst) &\text{otherwise}\\
		\end{cases}
	\end{align*}
\end{definition}

\begin{figure}[h]
	\scalebox{0.85}{
	\begin{mathpar}
		\inferrule[exec-input]{\mathit{min}(D) > t \\ (M, \mathit{stmt}) \evalStmt[\top] M'}{((M,D),\mathit{stmt}) \Rightarrow_{I, t} (M', D)}\qquad
		\inferrule[exec-deadline]{\mathit{min}(D) \le t \\ (M, \mathit{stmt}) \Downarrow^{\top}_{\emptyset,D,\mathit{min}(D)} M' \\\\ D' = \nextDlFn(D, min(D), M, M' )\\\\((M', D'), stmt) \Rightarrow_{I,t} (M'', D'')}{((M,D),\mathit{stmt}) \Rightarrow_{I, t} (M'', D'')}\\
	\end{mathpar}}
	\caption{Semantics of the evaluation of a step in the monitor}
	\label{fig:inf-rules:step}	
\end{figure}

\begin{definition}[Step semantics of statements]
	Given a memory $M$, an input $I$, and a timestamp $t$, a step of a statement $(M, \mathit{stmt}) \Rightarrow_{I, t} M'$ returns a memory $M'$ using the inference rule in \Cref{fig:inf-rules:step}.
\end{definition}

\begin{definition}[Semantics]
	Given an initial memory $M_0$ and an infinite input sequence $\inputs$, a monitor describes a sequence of \emph{memories}
	\begin{align*}
		\llbracket \syn{loop \{}stmt\syn{\}} \rrbracket_{\inputs, M_0} := &\{ M_0, M_1, \ldots \mid \forall n. M_n \Rightarrow_{\inputs[n]} M_{n+1}\}
	\end{align*}
  \end{definition}

  \begin{definition}[Well-definedness]
	A program $stmt$ with an initial memory $M$ is well-defined if, for any infinite sequence of inputs $\inputs$ with monotone time\-stamps, the program has exactly one evaluation model:
	\[
		\forall \inputs. \forall n. \mathit{Time}(\inputs[n]) < \mathit{Time}(\inputs[n + 1]) \rightarrow  \llbracket \syn{loop \{}stmt\syn{\}} \rrbracket_{\inputs, M} = \{ \world_\mathit{out}\}
	\]
\end{definition}

\section{Translation from \rtlola to \streamir}
\label{app:translation}
\renewcommand{\boxed}[1]{#1}
\begin{figure}[h]
	\centering
	\scalebox{0.9}{\parbox{\linewidth}{
	\begin{align*}
		\translate(\epsilon) &:= \syn{skip}\\
		\translate(layer \cdot xs) &:= \translate(layer)~\syn{;}~\translate(xs)\\
		\translate(\{task_0, ... task_n\}) &:= \translate(task_0)~\syn{$\|$}~\dots~\syn{$\|$}~\translate(task_n)\\
		\translate(Input(i)) &:= \boxed{\begin{aligned}
			&\syn{if}~\syn{inputs?}~i~\syn{then}~\syn{shift}~i~\syn{;}~\syn{input}~i\\
		\end{aligned}}\\
		\translate(Spawn(o)) &:= 
		\boxed{\begin{aligned}
			&\syn{if}~p_s~\syn{$\land$}~\syn{dynamic}~c_s~\syn{then}~\syn{spawn}~o~e_s
		\end{aligned}}\\
		\translate(Shift(o)) &:= 
		\boxed{\begin{aligned}
			&\syn{iterate}~o\\[-1.5mm]
			&\quad\syn{if}~p_e~\syn{$\land$}~\syn{dynamic}~c_e~\syn{then}~\syn{shift}~o\\
		\end{aligned}}\\
		\translate(Eval(o)) &:= 
		\boxed{\begin{aligned}
			&\syn{iterate}~o\\[-1.5mm]
			&\quad\syn{if}~p_e~\syn{$\land$}~\syn{dynamic}~c_e~\syn{then}~\syn{eval}~o~e\\
		\end{aligned}}\\
		\translate(Close(o)) &:= 
		\boxed{\begin{aligned}
			&\syn{iterate}~o\\[-1.5mm]
			&\quad\syn{if}~p_c~\syn{$\land$}~\syn{dynamic}~c_c~\syn{then}~\syn{close}~o\\
		\end{aligned}}
	\end{align*}}}
	\caption{Translation of an \rtlola specification to a monitor}
	\label{fig:tranlate}
\end{figure}

\begin{definition}[Translation]
	Given an \rtlola specification $\varphi$ with fully annotated output streams
	\begin{lstlisting}[frame=none,numbers=none]
		output $o$($inst$)
			spawn @$p_s$ when $c_s$ with $e_s$ eval @$p_e$ when $c_e$ with $e_e$ close @$p_c$ when $c_c$
	\end{lstlisting}
	and a partial order $\mathit{Layers}$ generated from the Dependency Graph of $\varphi$, the translation $\mathit{translate}_{\varphi}(\mathit{Layers})$ defined in \Cref{fig:tranlate} returns a statement $\mathit{stmt}$ where $\syn{loop \{}~\mathit{stmt}~\syn{\}}$ is a monitor $\psi$ for $\varphi$.
\end{definition}

\begin{definition}[Initial Memory]
	Given an \rtlola specification, the initial memory is defined as the pair $(M_0, D_0)$ with 
	\begin{align*}
		\forall sr \in \srin. M_0[sr][\top] &= \epsilon\\
		\forall sr \in \srout. \forall \mathit{inst} \in \VV. M_0[sr][\mathit{inst}] &= \bot\\
		\forall Global(\mathit{freq}) \in \varphi. D_0(freq) &= freq\\
		\forall Local(\mathit{freq}) \in \varphi. sr \in \srout. \mathit{inst} \in \VV. D_0(\mathit{freq}, sr, \mathit{inst}) &= \bot
	\end{align*}
\end{definition}

\begin{theorem}[Well-Definedness]
	For an \rtlola specification $\varphi$ and a partial order derived from the dependency graph of $\varphi$, the translation returns a well-defined monitor $\psi$ with initial memory $M_0$.
\end{theorem}
\begin{proof}
   The proof iterates over all failable inference rules for statements, expressions, and guards, and contradicts these rules to the well-definedness or well-typedness of $\varphi$, the list of layers or the initial memory to show the existence of an output.
   To prove that this output is unique, the proof uses the well-definedness of $\varphi$ and shows that the order of the instance iterations does not affect the computation, i.e., otherwise,  $\varphi$ would not be well-typed.	
\end{proof}

\section{Rewriting Rules}
\label{app:rewriting}

\newcommand{\rewriteRule}[3]{
	\hline\multicolumn{2}{l}{\textsc{#1}}\\\hline
	$\begin{aligned}#2\end{aligned}$ & $\begin{aligned}#3\end{aligned}$\\
}

\newcommand{\rewriteRuleIf}[4]{
	\hline\multicolumn{2}{l}{\textsc{#1}: #2}\\\hline
	$\begin{aligned}#3\end{aligned}$ & $\begin{aligned}#4\end{aligned}$\\
}

\subsection{General}

\begin{tabularx}{\linewidth}{X|X}
	\rewriteRuleIf{Common-If}{\syn{$\circ$} $\in$ \{~\syn{;}~,~\syn{$\|$}~\}}{
		&\syn{if}~c\land c_1~\syn{then}~A~\syn{$\circ$}~\syn{if}~c\land c_2~\syn{then}~B
	}{
		&\syn{if}~c~\syn{then}\\
		&\quad\syn{if}~c_1~\syn{then}~A~\syn{$\circ$}~\syn{if}~c_2~\syn{then}~B\\
	}
	\rewriteRuleIf{Nested-If}{\syn{$\circ$} $\in$ \{~\syn{;}~,~\syn{$\|$}~\}}{
		&(\syn{if}~c~\syn{then}~A)~\syn{$\circ$}~(\syn{if}~c~\syn{then}~B)
	}{
		&\syn{if}~c~\syn{then}~(A~\syn{$\circ$}~B)
	}
	\rewriteRule{Split-If}{
		&\syn{if}~c_1 \land c_2~\syn{then}~A
	}{
		&\syn{if}~c_1 ~\syn{then}~\syn{if}~c_2~\syn{then}~A
	}
	\rewriteRule{Combine-If}{
		&\syn{if}~c_1 ~\syn{then}~\syn{if}~c_2~\syn{then}~A
	}{
		&\syn{if}~c_1 \land c_2~\syn{then}~A
	}
	\rewriteRuleIf{Implied-If}{If $c_1 \Rightarrow c_2$}{
		&\syn{if}~c_1 ~\syn{then}~\ldots \syn{if}~c_2 ~\syn{then}~A~\ldots
	}{
		&\syn{if}~c_1~\syn{then}~\ldots A \ldots
	}
	\rewriteRuleIf{Associativity}{\syn{$\circ$} $\in$ \{~\syn{;}~,~\syn{$\|$}~\}}{
		&A~\syn{$\circ$}~(~B~\syn{$\circ$}~C~)
	}{
		&(~A~\syn{$\circ$}~B~)~\syn{$\circ$}~C
	}
	\rewriteRule{Swap-Par}{
		&A~\syn{$\|$}~B
	}{
		&B~\syn{$\|$}~A
	}
	\rewriteRuleIf{Iterate-Inside}{if $c$ does not contain \syn{self} and is not \syn{dynamic}}{
		&\syn{iterate}~o~(~\syn{if}~c~\syn{then}~A~)
	}{
		&\syn{if}~c~(~\syn{iterate}~o~A~)
	}
	\rewriteRuleIf{Unique-Assign}{if $c$ uniquely determines a parameter $p$}{
		&\syn{iterate}~o~(~\syn{if}~c~\syn{then}~A~)
	}{
		&\syn{assign}~p~(~\syn{if}~c~\syn{then}~A~)
	}
	\rewriteRuleIf{Combine-Iterate-Seq}{if $o1$ and $o2$ share spawn and close condition}{
		&(~\syn{iterate}~o1~A~)~\syn{;}~(\syn{iterate}~o2~B)
	}{
		&\syn{iterate}~o1~(A~\syn{;}~B~)
	}
	\hline
\end{tabularx}

\noindent
\begin{tabularx}{\linewidth}{X|X}
	\rewriteRuleIf{Combine-Iterate-Par}{if $o1$ and $o2$ share spawn and close condition}{
		&(~\syn{iterate}~o1~A~)~\syn{$\|$}~(\syn{iterate}~o2~B)
	}{
		&\syn{iterate}~o1~(A~\syn{$\|$}~B~)
	}
	\rewriteRuleIf{Single-Instance}{if $o$ has only a single instance}{
		&(~\syn{iterate}~o~A~)
	}{
		&A
	}
	\rewriteRule{If-True}{
		&\syn{if}~true~\syn{then}~A~\syn{else}~B
	}{
		&A
	}
	\rewriteRule{If-False}{
		&\syn{if}~false~\syn{then}~A~\syn{else}~B
	}{
		&B
	}
	\rewriteRuleIf{Remove-Skip-Seq}{\syn{$\circ$} $\in$ \{~\syn{;}~,~\syn{$\|$}~\}}{
		&\syn{A}~\syn{$\circ$}~\syn{skip}\quad\mid\quad\syn{skip}~\syn{$\circ$}~\syn{A}
	}{
		&\syn{A}
	}
	\rewriteRule{Remove-Skip-If}{
		&\syn{if}~c~\syn{skip}
	}{
		&\syn{skip}
	}
	\rewriteRule{Remove-Skip-Iterate}{
		&\syn{iterate}~o~\syn{skip}
	}{
		&\syn{skip}
	}
	\rewriteRule{Remove-Skip-Assign}{
		&\syn{assign}~p~\syn{skip}
	}{
		&\syn{skip}
	}
	\hline
\end{tabularx}

\subsection{Implementation Specific}

The following rewrite rules are based on assumptions about the implementation, such as the bounded memory being realized as a ring buffer:

\vspace{3mm}\noindent
\begin{tabularx}{\linewidth}{X|X}
	\rewriteRuleIf{Unecessary-Shift}{if $o$ has a memory bound $\le 1$}{
		&\syn{shift}~o
	}{
		&\syn{skip}
	}
	\rewriteRuleIf{Unecessary-Spawn}{if $o$ has a single instance without spawn condition}{
		&\syn{spawn}~o
	}{
		&\syn{skip}
	}
	\hline
\end{tabularx}

\subsection{Guards}

\vspace{3mm}\noindent
\begin{tabularx}{\linewidth}{X|X}
	\rewriteRule{And-True}{
		&c~\land~true\;|\;true~\land~c
	}{
		&c
	}
	\rewriteRule{And-False}{
		&c~\land~false\;|\;false~\land~c
	}{
		&false
	}
	\rewriteRule{Or-True}{
		&c~\lor~true\;|\;true~\lor~c
	}{
		&true
	}
	\rewriteRule{Or-False}{
		&c~\lor~false\;|\;false~\lor~c
	}{
		&c
	}
	\rewriteRuleIf{Dynamic-Op}{$\circ$~$\in$~\{~$\land$~,~$\lor$~\}}{
		&\syn{dynamic}~(c_1~\circ~c_2)
	}{
		&\syn{dynamic}~c_1~\circ~\syn{dynamic}~c_2
	}
	\hline
\end{tabularx}

\section{Contracts}
\label{app:contracts}

For the evaluation, we only list functions updating the state of the contract.
Getter functions are ommited.

\subsubsection{Order Contract~\cite{azzopardi2018monitoring,Scheerer/21}} \begin{enumerate}
	\item \emph{acceptContract:} Both parties, the buyer and seller, need to accept the contract before the placement of the first order.
	\item \emph{deliveryMade:} Marks an order as delivered by the seller.
	\item \emph{escrow:} Before accepting the contract, the buyer has to pay an escrow of a minimum number of items, while the seller has to pay an performance guarantee.
	\item \emph{paymentReceived:} The seller marks the payment of an order as received.
	\item \emph{placeOrder:} The buyer places a new order.
	\item \emph{requestTermination:} Either party can request termination of the contract.
\end{enumerate}

\subsubsection{ERC20 Contract~\cite{ERC20Reference}} \begin{enumerate}
\item \emph{approve:} Permit another account to transfer tokens from this account.
\item \emph{transfer:} Transfer tokens from your account to someone elses.
\item \emph{transferFrom:} Transfer tokens from another account, if it was approved.
\end{enumerate}

\subsubsection{Voting Contract~\cite{VotingReference}} \begin{enumerate}
\item \emph{endVoting:} End the voting phase.
\item \emph{giveRightToVote:} Assign someone the right to vote.
\item \emph{startVoting:} Start the voting phase.
\item \emph{voteFor:} Cast your vote.
\end{enumerate}

\subsubsection{Crowdfunding Contract~\cite{CrowdFundingReference}} \begin{enumerate}
\item \emph{bid:} Contribute money to the fund.
\item \emph{claim:} If goal was reached, claim the funded money.
\item \emph{refund:} If goal was not reached, refund your contribution.
\end{enumerate}

\subsubsection{Auction Contract~\cite{SimpleAuctionReference}} \begin{enumerate}
\item \emph{bid:} Bit to the auction.
\item \emph{end:} Mark the auction as ended.
\item \emph{withdraw:} Withdraw money if previous bid was overbid.
\end{enumerate}

\section{Evaluation Results}
\label{app:runtimes}
The following tables displays the mean runtime over 20 runs across the different parameters.
They correspond to the runtimes depicted in figures \Cref{fig:eval:geofence} and \Cref{fig:eval:intruder} respectively.

\subsection{Geofence Specification}
\setlength{\tabcolsep}{8pt}
\begin{center}
\begin{tabular}{rrrrrr}
\toprule
parameter & \multicolumn{5}{c}{Runtime [ms]}\\\cmidrule(lr){2-6}
& ip\_old & ip\_new & ip\_new\_opt & rustc & rustc\_opt \\
\midrule
1 & 25.7 & 21.0 & 12.8 & 7.8 & 6.2 \\
2 & 29.4 & 25.3 & 13.3 & 7.2 & 6.7 \\
3 & 33.1 & 29.0 & 14.1 & 5.6 & 5.2 \\
4 & 36.7 & 32.6 & 14.5 & 7.1 & 7.1 \\
5 & 40.2 & 36.2 & 15.1 & 5.8 & 5.3 \\
6 & 43.7 & 39.5 & 15.7 & 7.3 & 6.8 \\
7 & 47.3 & 42.9 & 16.3 & 6.8 & 5.6 \\
8 & 50.7 & 45.8 & 16.8 & 6.1 & 5.3 \\
9 & 54.2 & 48.7 & 17.3 & 6.8 & 5.9 \\
10 & 57.9 & 51.9 & 18.1 & 6.4 & 5.6 \\
11 & 61.3 & 54.9 & 18.5 & 7.4 & 6.6 \\
12 & 64.8 & 58.1 & 19.0 & 7.2 & 5.7 \\
13 & 68.3 & 61.4 & 19.6 & 8.4 & 6.8 \\
14 & 72.1 & 64.5 & 20.2 & 7.0 & 5.4 \\
15 & 75.4 & 67.7 & 20.7 & 8.5 & 7.1 \\
\bottomrule
\end{tabular}
\end{center}

\subsection{Intruder Specification}

\begin{center}
\begin{tabular}{rrrrrr}
\toprule
parameter & \multicolumn{5}{c}{Runtime [ms]}\\\cmidrule(lr){2-6}
& ip\_old & ip\_new & ip\_new\_opt & rustc & rustc\_opt \\
\midrule
50 & 135.9 & 98.1 & 26.6 & 17.5 & 8.2 \\
100 & 222.2 & 162.4 & 26.9 & 27.5 & 9.9 \\
150 & 312.8 & 226.4 & 27.0 & 38.3 & 11.7 \\
200 & 396.7 & 286.7 & 27.2 & 45.8 & 13.9 \\
250 & 484.6 & 349.5 & 27.2 & 58.0 & 15.8 \\
300 & 574.5 & 406.5 & 27.3 & 65.9 & 19.9 \\
350 & 663.2 & 465.8 & 27.7 & 74.9 & 22.3 \\
400 & 757.1 & 519.6 & 27.8 & 84.3 & 24.5 \\
450 & 851.2 & 587.0 & 27.4 & 95.0 & 29.0 \\
500 & 929.6 & 641.6 & 27.4 & 104.8 & 31.3 \\
\bottomrule
\end{tabular}
\end{center}
}{}

\end{document}